\documentclass{elsarticle}

\journal{Journal of Logic and Computation}

\usepackage[english]{babel}
\usepackage[T1]{fontenc}
\usepackage[latin1]{inputenc}
\usepackage{times}
\usepackage{amsfonts,amsthm,amssymb,amsmath,amscd,latexsym,graphicx}

\newtheorem{conjecture}{Conjecture}
\newtheorem{definition}{Definition}
\newtheorem{theorem}{Theorem}
\newtheorem{lemma}{Lemma}

\newtheorem{corollary}{Corollary}

\newcommand{\set}[1]{\left\{ #1 \right\}}

\newcommand{\prob}[1]{\mathbb{P}_{ #1 }}
\newcommand{\Lx}[2]{L^{> #1}(#2)}
\newcommand{\Lhalf}[1]{\Lx{\frac{1}{2}}{#1}}
\newcommand{\bin}{\mathrm{bin}}

\newcommand{\D}{\mathcal{D}}
\newcommand{\G}{\mathcal{G}}

\newcommand{\Alt}[1]{\mathrm{Alt}\left( #1 \right)}
\newcommand{\ND}[1]{\mathrm{NonDet}\left( #1 \right)}
\newcommand{\Det}[1]{\mathrm{Det}\left( #1 \right)}
\newcommand{\Reg}{\mathrm{Reg}}

\newcommand{\Aut}{\mathrm{Aut}}

\renewcommand{\L}{\mathcal{L}}

\newcommand{\Eq}{\textsc{CountEq}}
\newcommand{\NotEq}{\textsc{NotEq}}
\newcommand{\Leq}{\textsc{Lexicographic}}
\newcommand{\Prime}{\textsc{Primes}}

\newcommand{\M}{\mathcal{M}}

\newcommand{\N}{\mathbb{N}}
\newcommand{\Z}{\mathbb{Z}}
\newcommand{\A}{\mathcal{A}}
\newcommand{\B}{\mathcal{B}}

\begin{document}

\begin{frontmatter}

\title{Lower bounds for the state complexity of probabilistic languages and the language of prime numbers}
\tnotetext[mytitlenote]{This journal version extends two conference papers: 
the first published in the proceedings of LFCS'2016~\cite{Fijalkow16}, 
and the second published in the proceedings of LICS'2018~\cite{Fijalkow18}.}

\author{Nathana\"el Fijalkow}
\address{CNRS, LaBRI, Bordeaux, and The Alan Turing Institute of data science and artificial intelligence, London}

\begin{abstract}
This paper studies the complexity of languages of finite words using automata theory.
To go beyond the class of regular languages, we consider \textit{infinite} automata and the notion of \textit{state complexity} defined by Karp.

Motivated by the seminal paper of Rabin from 1963 introducing probabilistic automata, we study the (deterministic) state complexity of probabilistic languages and prove that probabilistic languages can have arbitrarily high deterministic state complexity.

We then look at \textit{alternating automata} as introduced by Chandra, Kozen and Stockmeyer:
such machines run independent computations on the word and gather their answers through boolean combinations.
We devise a lower bound technique relying on boundedly generated lattices of languages, and give two applications of this technique.
The first is a hierarchy theorem, stating that there are languages of arbitrarily high polynomial alternating state complexity,
and the second is a linear lower bound on the alternating state complexity of the prime numbers written in binary.
This second result strengthens a result of Hartmanis and Shank from 1968, which implies an exponentially worse lower bound for the same model.
\end{abstract}

\begin{keyword}
State Complexity \sep Automata \sep Alternating Automata \sep Probabilistic Languages \sep Complexity Theory \sep Prime Numbers
\end{keyword}

\end{frontmatter}

\section{Introduction}
\label{sec:intro}
The seminal paper of Karp~\cite{Karp67} defines the \textit{state complexity} of an (infinite) automaton
as a function associating with $n$ the number of states reachable by reading a word of length at most $n$.
For a function $f : \N \to \N$, a language $L \subseteq A^*$ has state complexity $f$ if 
there exists an automaton recognising $L$ of state complexity at most $f$.

\vskip1em
For the case of deterministic automata, the state complexity is fully characterised by the celebrated Myhill-Nerode theorem~\cite{Nerode58},
which states the existence of a canonical minimal (potentially infinite) automaton for a given language based on the notion of left quotients.
Nevertheless, it is sometimes complicated to understand the structure of this automaton,
as demonstrated by the case of the language of prime numbers written in binary:
a series of papers culminates in a result of Hartmanis and Shank~\cite{HartmanisShank69} showing that this language 
has asymptotically maximal (i.e., exponential) deterministic state complexity.

\vskip1em
Our first aim is to investigate the deterministic state complexity of probabilistic automata, 
a simple probabilistic model of computation introduced by Rabin in his seminal paper~\cite{Rabin63}.
This study is motivated by the section ``approximate calculation of matrix products'' in this paper;
in the end of this section, Rabin states a result, without proof; 
we substantiate this claim, \textit{i.e.} formalise and prove the result.

\vskip1em
We then initiate the study of \textit{alternating state complexity}, which uses Karp's definition
instantiated with (infinite) alternating automata.
We first motivate the model with some examples and later discuss its relevance.
Formal definitions are given in the next section; we stick to intuitive explanations in this introduction.

\vskip1em
Consider the language 
\[
\Eq_3 = \set{w \in \set{a,b,c}^* \mid |w|_a = |w|_b = |w|_c},
\]
consisting of words having the same number of $a$'s, $b$'s and $c$'s.
(We let $|w|_a$ denote the number of letters $a$ in $w$.)
This language is not regular, but we claim that it is recognised by a deterministic automaton of quadratic state complexity.
Indeed, we construct an automaton whose set of states is $\Z^2$, interpreted as two counters.
They are initialised to $0$ each and maintain the value $(|w|_a - |w|_b, |w|_a - |w|_c)$.
To this end, the letter $a$ acts as $(+1,+1)$,
the letter $b$ as $(-1,0)$, the letter $c$ as $(0,-1)$.
The only accepting state is $(0,0)$.
This automaton is of quadratic state complexity: after reading the word $w$ the automaton is in the state $(|w|_a - |w|_b, |w|_a - |w|_c)$,
which means that the set of states reachable by words of length at most $n$ has size $(2n+1)^2$.

\vskip1em
Consider now the language
\[
\NotEq = \set{u \sharp v \mid u,v \in \set{0,1}^*, u \neq v},
\]
consisting of two words $u,v$ over the alphabet $\set{0,1}$ separated by the letter $\sharp$
such that $u$ is different from $v$.
One can easily see that this language does not have subexponential deterministic state complexity:
after reading two different words $u$ and $u'$, any deterministic automaton recognising $\NotEq$
must be in two different states.

However, it is recognised by a non-deterministic automaton of linear state complexity.
Note that there are three ways to have $u \neq v$: either $v$ is longer than $u$,
or $v$ is shorter than $u$, or there exists a position at which they differ.
At the beginning the automaton guesses which of these three situations occur.
We focus on the third possibility for the informal explanation.
The automaton guesses a position in the first word, stores in the state the position $p$ together with the letter $a$ at this position,
and checks whether the corresponding position in the second word indeed differs.
To this end, after reading the letter $\sharp$, it decrements the position until reaching~$0$, and checks whether the letter is indeed different
than the letter stored in the state.

\vskip1em
Our third example is the language
\[
\Leq = \set{u \sharp v \mid u,v \in \set{0,1}^*, u <_{\text{lex}} v},
\]
consisting of two words $u,v$ over the alphabet $\set{0,1}$ separated by the letter~$\sharp$
such that $u$ is lexicographically smaller than $v$.
One can see that this language does not have subexponential non-deterministic state complexity 
(we do not substantiate this claim here).
However, we will now explain that it is recognised by an alternating automaton of linear state complexity.

The notion of alternating (Turing) machines was introduced by Chandra, Kozen and Stockmeyer~\cite{ChandraStockmeyer76,Kozen76,ChandraKozenStockmeyer81}.
A non-deterministic automaton makes guesses about the word, and the computation is accepting 
if there exists a sequence of correct guesses. 
In other words, these guesses are disjunctive choices; 
the alternating model restores the symmetry by introducing disjunctive and conjunctive choices.
Whenever the automaton makes a choice, we say that it creates independent copies of itself, one for each alternatives;
if the choice was disjunctive, the computation is accepted if some copy accepts, and if the choice was conjunctive,
the computation is accepted if all copies accept.

We illustrate this notion by constructing an alternating automaton for the language $\Leq$.
We unravel the inductive definition of the lexicographic order: 
$u <_{\text{lex}} v$ if and only if
\[
\left(u(0) = 0 \wedge v(0) = 1\right) \vee \left(u(0) = v(0) \wedge u(\ge 1) <_{\text{lex}} v(\ge 1)\right).
\]
Here $u(0)$ is the first letter of $u$, and $u(\ge 1)$ is the word $u$ stripped of its first letter.
Upon reading the first letter $u(0)$, the automaton makes a disjunctive guess corresponding to the disjunction in the definition: 
either both $u(0) = 0$ and $v(0) = 1$, 
or both $u(0) = v(0)$ and $u(\ge 1) <_{\text{lex}} v(\ge 1)$.
In the latter case, the automaton makes a further choice, conjunctive this time,
checking with one copy that $u(0) = v(0)$ and with another that $u(\ge 1) <_{\text{lex}} v(\ge 1)$.

\vskip1em
\textbf{Alternating automata are succinct.}
It is well-known that finite deterministic, non-deterministic and alternating automata are equivalent.
As hinted by the examples discussed above, for infinite automata we do not have such an equivalence.
Some classical constructions still apply, for instance the powerset construction to determinise automata
which increases the state complexity exponentially. Similarly one can transform alternating automata
into deterministic ones increasing the state complexity by a two-fold exponential.
Hence one can see alternating automata as a class of \textit{succinctly} represented deterministic automata,
whose inner boolean structure is made explicit.

\vskip1em
\textbf{Alternating automata are distributed.}
Another appeal of alternating automata is as a model of distributed computation. 
Indeed, in the course of its computation, an alternating automaton produces copies of itself that can be run independently
on a distributed architecture.
The final output is then computed by boolean combinations of the answers of each copy.
This point of view echoes the recent work of Reiter~\cite{Reiter15}, 
which combines ideas from distributed algorithms and alternating automata.

\vskip1em
\textbf{Applications.}
The notion of state complexity is used as a complexity measure to evaluate how complicated
some operations on languages are. We refer to the surveys~\cite{Yu01,Yu02,GMRY17}
for more details on this long line of work.
The other natural use of state complexity is as a tool for separating models of computations.
For instance, the paper of Dawar and Kreutzer~\cite{DK07} generalises the notion of automaticity (see related works)
to relational structures and uses it for separating several modal and non-modal fixed-point logics.

\vskip1em
\textbf{Contributions of the paper.} We devise a generic lower bound technique for alternating state complexity 
based on boundedly generated lattices of languages.

We give the basic definitions and show some examples in Section~\ref{sec:defs}.
The Section~\ref{sec:claim} is devoted to substantiating Rabin's claim about the deterministic state complexity
of probabilistic languages.
We discuss related works in Section~\ref{sec:related_works}.
We describe our lower bound technique in Section~\ref{sec:lower_bound}, and give two applications:
\begin{itemize}
	\item \textbf{Hierarchy theorem}: in Section~\ref{sec:hierarchy}, we prove a hierarchy theorem: 
	for each natural number $\ell$ greater than or equal to $2$, 
	there exists a language having alternating state complexity $n^{\ell}$ but not $n^{\ell - \varepsilon}$ for any $\varepsilon > 0$.

	\item \textbf{Prime numbers}: in Section~\ref{sec:prime}, we look at the language of prime numbers written in binary.
	The works of Hartmanis and Shank culminated in showing that it does not have subexponential \textit{deterministic} state complexity~\cite{HartmanisShank69}.
	We consider the stronger model of \textit{alternating} automata, 
	and first observe that Hartmanis and Shank's techniques imply a \textit{logarithmic} lower bound 
	on the \textit{alternating} state complexity.
	Our contribution is to strengthen this result by showing a \textit{linear} lower bound, which is thus an exponential improvement.
\end{itemize}

\section{Definitions}
\label{sec:defs}

\subsection{State Complexity}
\label{subsec:sc}
We fix an \textit{alphabet} $A$, which is a finite set of letters.
A \textit{word} is a finite sequence of letters $w = w(0) w(1) \cdots w(n-1)$, 
where the $w(i)$ are letters from the alphabet $A$, i.e., $w(i) \in A$.
We say that $w$ has length $n$, and write $|w|$ for the length of $w$.
The empty word is $\varepsilon$.
We let $A^*$ denote the set of all words and $A^{\le n}$ the set of words of length at most $n$.
A language, typically denoted by $L$, is a set of words.

For a set $E$, we let $\B^+(E)$ denote the set of boolean formulae over $E$,
i.e., using conjunctions and disjunctions.
Throughout the paper we only consider \emph{positive} boolean combinations.
For instance, if $E = \set{p,q,r}$, an element of $\B^+(E)$ is $p \wedge (q \vee r)$.
A conjunctive formula uses only conjunctions, and a disjunctive formula only disjunctions.
For $\delta \in \B^+(E)$ and $X \subseteq E$, we write $X \models \delta$
if $\delta$ is true when setting the elements of $X$ to true and the others to false.

\begin{definition}[Alternating Automata~\cite{ChandraStockmeyer76,Kozen76,ChandraKozenStockmeyer81}]
An alternating automaton is given by a (potentially infinite) set $Q$ of states, an initial state $q_0 \in Q$,
a transition function $\delta : Q \times A \to \B^+(Q)$ and a set of accepting states $F \subseteq Q$.
\end{definition}

We use acceptance games to define the semantics of alternating automata.
Consider an alternating automaton $\A$ and an word $w$, 
we define the acceptance game $\G_{\A,w}$ as follows: it has two players, Eve and Adam.
Eve claims that the word $w$ should be accepted, and Adam challenges this claim.

The game starts from the initial state $q_0$, and with each letter of $w$ read from left to right,
a state is chosen through the interaction of the two players.
If in a state $q$ and reading a letter $a$, 
Eve and Adam look at the boolean formula $\delta(q,a)$; 
Eve chooses which clause is satisfied in a disjunction, and Adam does the same for conjunctions.
This leads to a new state $p$, from which the computation continues.
A play is won by Eve if it ends up in an accepting state.

The word $w$ is accepted by $\A$ if Eve has a winning strategy in the acceptance game $\G_{\A,w}$.
The language recognised by $\A$ is the set of words accepted by $\A$.

\vskip1em
As special cases, an automaton is
\begin{itemize}
	\item \textit{non-deterministic} if for all $q$ in $Q$, $a$ in $A$, $\delta(q,a)$ is a disjunctive formula,
	\item \textit{universal} if for all $q$ in $Q$, $a$ in $A$, $\delta(q,a)$ is a conjunctive formula,
	\item \textit{deterministic} if for all $q$ in $Q$, $a$ in $A$, $\delta(q,a)$ is an atomic formula,
	i.e., if $\delta : Q \times A \to Q$.
\end{itemize}

\begin{definition}[State Complexity Classes~\cite{Karp67}]
Fix a function $f : \N \to \N$.
The language $L$ is in $\Alt{f}$ if there exists an alternating automaton recognising $L$ and a constant $C$ such that for all $n$ in $\N$:
\[
\left|\set{q \in Q \mid \exists w \in A^{\le n}, \text{ it is possible to reach } q \text{ in } \G_{\A,w}}\right| \le C \cdot f(n).
\]
\noindent Similarly, we define $\ND{f}$ for non-deterministic automata and $\Det{f}$ for deterministic automata.
\end{definition}

We write $f(n)$ for the function $f : n \mapsto f(n)$, so for instance $\Alt{n}$ is the class of languages having linear alternating state complexity.
We say that $L$ has sublinear (respectively subexponential) alternating state complexity if it is recognised by an alternating automaton of state complexity at most $f$,
where $f = o(n)$ (respectively $f = 2^{o(n)}$).

\vskip1em
We let $\Reg$ denote the class of regular languages, i.e., those recognised by finite automata.
Then \[\Det{1} = \ND{1} = \Alt{1} = \Reg,\] i.e., a language has constant state complexity if and only if it is regular.

We remark that $\Det{|A|^n}$ is the class of all languages.
Indeed, consider a language $L$, we construct a deterministic automaton recognising $L$ of exponential state complexity.
Its set of states is $A^*$, the initial state is $\varepsilon$ and the transition function is defined by $\delta(w,a) = wa$.
The set of accepting states is simply $L$ itself.
The number of different states reachable by all words of length at most $n$ is the number of words of length at most $n$, 
i.e., $\frac{|A|^{n+1} - 1}{|A| - 1}$.

It follows that the asymptotical maximal state complexity of a language is exponential,
and the state complexity classes are relevant for functions smaller than exponential.

\subsection{The Myhill-Nerode Theorem}
\label{subsec:myhill-nerode}
We present an equivalent point of view on the deterministic state complexity based on Myhill-Nerode equivalence relation.

Let $w$ be a finite word, define its left quotient with respect to $L$ by
$$w^{-1} L = \set{u \mid w u \in L}.$$

A well known result from automata theory states that for all regular languages,
there exists a minimal deterministic finite automaton, called the \textit{syntactic automaton}, whose states is the set of left quotients.

This construction extends \textit{mutatis mutandis} when dropping the assumption that the automaton has finitely many states.
The statement gives precise lower bounds on the deterministic state complexity of the language.

Formally, consider a language $L$, we define the syntactic automaton of $L$, denoted $\A_L$, as follows. 
We define the set of states as the set of all left quotients: $\set{w^{-1} L \mid w \in A^*}$.
The initial state is $\varepsilon^{-1} L$, and the transition function is defined by
$\delta(w^{-1} L, a) = (wa)^{-1} L$.
Finally, the set of accepting states is $\set{w^{-1} L \mid w \in L}$.

Let $f_L : \N \to \N$ defined by 
\[
f_L(n) = |\set{w^{-1} L \mid w \in A^{\le n}}|.
\]

\begin{theorem}[Reformulation of Myhill-Nerode Theorem~\cite{Nerode58}]\leavevmode
\label{thm:lower_bound_left_quotients}
\begin{itemize}
	\item $\A_L$ recognises $L$, so $L \in \Det{f_L}$,
	\item for all $f$, if $L \in \Det{f}$, then $f = \Omega(f_L)$.
\end{itemize}
\end{theorem}

The first item is routinely proved.
For the second item, we prove an even stronger property.
Assume towards contradiction that there exists an automaton of state complexity $f$ recognising $L$
and such that there exists $n$ such that $f(n) < f_L(n)$.
Since $f(n) < f_L(n)$, there exists two words $u$ and $v$ of length at most $n$
such that $u^{-1} L \neq v^{-1} L$ but in $\A$ the words $u$ and $v$ lead to the same state.
The left quotients $u^{-1} L \neq v^{-1} L$ being different, there exists a word $w$ such that 
$uw \in L$ and $vw \notin L$, or the other way around.
But since the words $u$ and $v$ lead to the same state and $\A$ is deterministic, 
this state must be both accepting and rejecting, contradiction.

\subsection{Probabilistic Automata}
\label{subsec:pa}
Let $Q$ be a finite set of states.
A distribution over $Q$ is a function $\delta : Q \to [0,1]$ such that $\sum_{q \in Q} \delta(q) = 1$.
We denote $\D(Q)$ the set of distributions over $Q$.

\begin{definition}[Probabilistic Automaton]
A \textit{probabilistic automaton} $\A$ is given by a finite set of states $Q$,
a transition function $\phi : A \to (Q \to \D(Q))$,
an initial state $q_0 \in Q$, and a set of final states $F \subseteq Q$.
\end{definition}


In a transition function $\phi$, the quantity $\phi(a)(s,t)$ is the probability to go from the state $s$ to the state $t$ reading the letter $a$.
A transition function naturally induces a morphism $\phi : A^* \to (Q \to \D(Q))$.
We denote $\prob{\A}(s \xrightarrow{w} t)$ 
the probability to go from a state $s$ to a state $t$ reading $w$ 
on the automaton $\A$, \textit{i.e.}~$\phi(w)(s,t)$.
The \emph{acceptance probability} of a word $w \in A^*$ by $\A$ is 
$\sum_{t \in F} \phi(w)(q_0,t)$, which we denote $\prob{\A}(w)$.

The following threshold semantics was introduced by Rabin~\cite{Rabin63}.

\begin{definition}[Probabilistic Language]
Let $\A$ be a probabilistic automaton and $x$ a threshold in $(0,1)$, it induces the \textit{probabilistic language}
$$\Lx{x}{\A} = \set{w \in A^* \mid \prob{\A}(w) > x}.$$
\end{definition}

\section{Substantiating the Claim of Rabin}
\label{sec:claim}
In the section called ``approximate calculation of matrix products'' 
in the paper introducing probabilistic automata~\cite{Rabin63}, Rabin asks the following question: 
is it possible, given a probabilistic automaton, to construct an algorithm which reads words and compute the acceptance probability in an online fashion?

He first shows that this is possible under some restrictions on the probabilistic automaton,
and concludes the section by stating that ``\textit{an example due to R. E. Stearns shows that without assumptions, a computational procedure need not exist}''.
The example is not given, and to the best of the author's knowledge, has never been published anywhere.

\vskip1em
In this section we substantiate this claim using the framework of deterministic state complexity.
Whether this exactly fleshes out Rabin's claim is subject to discussions, since Rabin asks whether the acceptance probability 
can be computed up to a given precision; 
in our setting, the acceptance probability is not actually computed, but only compared to a fixed threshold, following Rabin's definition of probabilistic languages.

\vskip1em
The following result shows that there exists a probabilistic automaton defining a language of asymptotically maximal (exponential) deterministic state complexity.

\begin{theorem}
\label{thm:example}
There exists a probabilistic automaton $\A$ such that $\Lhalf{\A}$ does not have subexponential deterministic state complexity.
\end{theorem}

\begin{figure}[ht]
\begin{center}
\includegraphics[scale=1]{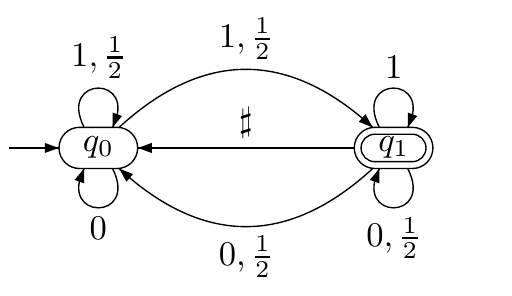}
\caption{\label{fig:example} 
The initial state is marked by an ingoing arrow and the accepting state by an outgoing arrow. The first symbol over a transition is a letter (either $0$, $1$, or $\sharp$). The second symbol (if given) is the probability of this transition. If there is only one symbol, then the probability of the transition is $1$.}
\end{center}
\end{figure}

In the original paper introducing probabilistic automata, Rabin~\cite{Rabin63} gave an example of a probabilistic automaton $\A$
computing the binary decomposition function (over the alphabet $\set{0,1}$), denoted $\bin$, \textit{i.e.} $\prob{\A}(u) = \bin(u)$,
defined by 
$$\bin(a_1 \ldots a_n) = \frac{a_1}{2^n} + \cdots + \frac{a_n}{2^1}$$
(\textit{i.e.} $0.a_n\ldots a_1$ in binary). 
We show that adding one letter and one transition to this probabilistic automaton induces a language 
which does not have subexponential deterministic state complexity.

\bigskip
The automaton $\A$ is represented in Figure~\ref{fig:example}.
The alphabet is $A = \set{0,1,\sharp}$.
The only difference between the automaton proposed by Rabin~\cite{Rabin63} and this one is the transition over $\sharp$ from $q_1$ to $q_0$.
As observed by Rabin, a simple induction shows that for $u$ in $\set{0,1}^*$, we have $\prob{\A}(u) = \bin(u)$.

Let $w \in A^*$, it decomposes uniquely into $w = u_1 \sharp u_2 \sharp \cdots \sharp u_k$, where $u_i \in \set{0,1}^*$.
Observe that $\prob{\A}(w) = \bin(u_1) \cdot \bin(u_2) \cdots \bin(u_k)$. 

\bigskip
Consider an automaton recognising $\Lhalf{\A}$ and fix $n$.
The binary decomposition function maps words of length $n$ to rationals of the form $\frac{a}{2^n}$, for $0 \le a < 2^n$.
Consider two different words $u$ and $v$ in $\set{0,1}^*$ of length $n$,
we show that $(u 1)^{-1} \Lhalf{\A} \neq (v 1)^{-1} \Lhalf{\A}$.

Without loss of generality assume $\bin(u 1) < \bin(v 1)$;
observe that $\frac{1}{2} \le \bin(u 1) < \bin(v 1)$.
There exists $w$ in $\set{0,1}^*$ such that $\bin(u 1) \cdot \bin(w) < \frac{1}{2}$ and $\bin(v 1) \cdot \bin(w) > \frac{1}{2}$:
it suffices to choose $w$ such that $\bin(w)$ is in $\left(\frac{1}{2\bin(v 1)},\frac{1}{2\bin(u 1)}\right)$,
which exists by density of the dyadic numbers in $(0,1)$.
Thus, $(u 1)^{-1} \Lhalf{\A} \neq (v 1)^{-1} \Lhalf{\A}$,
and we exhibited exponentially many words having pairwise distinct left quotients.

It follows from Theorem~\ref{thm:lower_bound_left_quotients} that $\Lhalf{\A}$ does not have subexponential deterministic state complexity.

We note that expanding on these ideas we gave a simple proof of the undecidability of the regularity problem for probabilistic languages~\cite{FS15}, which can be easily adapted to show that deciding the deterministic state complexity of a probabilistic language is undecidable.

\section{Related Works}
\label{sec:related_works}
The definition of state complexity is due to Karp~\cite{Karp67},
and the first result proved in that paper is that non-regular languages have at least linear deterministic state complexity.
Hartmanis and Shank considered the language of prime numbers written in binary, and showed in~\cite{HartmanisShank69} that 
it does not have subexponential deterministic state complexity.
We pursue this question in this paper by considering the alternating state complexity of the prime numbers.

\vskip1em
Automaticity was defined by Shallit and Breitbart and studied in depth in a series of four papers~\cite{ShallitBreitbart96,PomeranceRobsonShallit97,GlaisterShallit98,Shallit96}.

\begin{definition}
The \emph{automaticity} of a language $L$ is the function $\Aut(L) : \N \to \N$ which associates with $n$ the size of the smallest deterministic automaton
which agrees with $L$ on all words of length at most~$n$.
\end{definition}

The conceptual difference is that automaticity is a \textit{non-uniform notion}, 
since there is a finite automaton for each $n$, whereas state complexity is \textit{uniform}, 
since it considers one infinite automaton. 
For this reason, the two measures behave completely differently.

For instance, consider the language
\[
L_{\log} = \set{ w \in \set{a,b,\sharp}^* \left|
\begin{array}{ll}
w = uv \cdot \sharp \cdot u, \\
u,v \in \set{a,b}^*, |u| = \lfloor \log(|w|) \rfloor
\end{array} \right.
}.
\]
In words: the prefix of $w$ of length $\lfloor \log(|w|) \rfloor$ repeats just after the unique letter $\sharp$.

The automaticity of this language is linear, i.e., rather small: $\Aut(L_{\log})(n) = O(n)$.
Indeed, given $n$, the automaton $\A_n$ stores the prefix up to $\lfloor \log(n) \rfloor$, 
waits for the letter $\sharp$, and compares it to the word starting after $\sharp$.

On the other hand, the deterministic state complexity of $L_{\log}$ is asymptotically maximal, meaning exponential: 
indeed, since the automaton has no information on how long the prefix to be repeated may be, it has to store the whole word.
More formally, for any two words $u \neq v$, any deterministic automaton recognising $L_{\log}$ must be in two different states after reading $u$ and after reading $v$.

Note that replacing $\log$ by a very slow growing function yields examples showing that the gap between automaticity
and deterministic state complexity is arbitrarily large.

\vskip1em
Another interesting point to make here is the difference between finite and infinite automata.
Indeed, studying the state complexity of finite alternating automata can be reduced to the state complexity of finite deterministic automata
by reversing the words.
The notation $u^R$ stands for the reverse of $u$: \[u^R = u(n-1) \cdots u(0).\]
We extend it to languages: $L^R = \set{u^R \mid u \in L}$.
The following result is a variant of Brzozowski's minimization by reversal technique~\cite{B63},
and a classical result in automata theory.

\begin{lemma}[\cite{RS97,FJY90}]\hfill
\begin{itemize}
	\item If $L$ is recognised by an alternating automaton with $n$ states, then $L^R$
	is recognised by a deterministic automaton with $2^n$ states.
	\item If $L$ is recognised by a deterministic automaton with $2^n$ states, then $L^R$
	is recognised by an alternating automaton with $n$ states.
\end{itemize}
\end{lemma}
In other words, the number of states of the smallest finite deterministic automaton recognising $L$ is 
(almost) exactly $2^n$, where $n$ is the number of states of the smallest finite alternating automaton recognising $L^R$.

This result does not extend to state complexity for infinite automata: indeed, since every language has exponential deterministic state complexity,
this would imply that every language also has linear alternating state complexity.
That does not hold: we exhibit in Subsection~\ref{subsec:lower_bound_example} a language
which does not have subexponential alternating state complexity.

\vskip1em
Two notions share some features with alternating state complexity.

The first is boolean circuits; the resemblance is only superficial, as circuits do not process the input from left to right.
For instance, one can observe that the language \text{Parity}, which is hard to compute with a circuit (not in $\text{AC}^0$ for instance), 
is actually a regular language, so trivial with respect to state complexity.

\vskip1em
The second notion is alternating communication complexity, developed by Babai, Frankl and Simon~\cite{BabaiFranklSimon86}.
In this setting, Alice has an input $x$ in $A$, Bob an input $y$ in $B$, and they want to determine $h(x,y)$ for a given boolean function 
$h : A \times B \to \set{0,1}$ known by all.
Alice and Bob are referees in a discussion involving two individuals, Eve and Adam.
Eve tries to convince Alice and Bob that $h(x,y) = 1$, and Adam aims at the opposite.
A protocol of exchanging messages depending on the inputs is agreed upon by everyone beforehand.
Then the input $x$ is revealed to Alice and $y$ to Bob.
Eve and Adam both know the two inputs and exchange messages whose conformity to the inputs is checked by Alice and Bob.
The cost of the protocol is the number of bits exchanged.

The main difference between alternating communication complexity and state complexity is that protocols do not have to extract information
from the inputs sequentially as an automaton does.
For instance, swapping the inputs of Alice and Bob does not make any difference for communication complexity
but can completely change the state complexity.

As an example, consider the following language studied in Subsection~\ref{subsec:lower_bound_example}.
$$L = \set{u \sharp u_1 \sharp u_2 \sharp \cdots \sharp u_k 
\left| 
\begin{array}{c}
u,u_1,\ldots,u_k \in \set{0,1}^*, \\
\exists j \in \set{1,\ldots,k}, u = u_j
\end{array}
\right.}.$$
Alice receives $u$ of length $n$ and Bob receives $u_1 \sharp u_2 \sharp \cdots \sharp u_k$,
and they want to check whether there exists $j \in \set{1,\ldots,k}$ such that $u = u_j$.
A simple protocol is for Eve to send $j$, and then for Adam to send $i \in \set{1,\ldots,n}$
together with the letter $u(i)$, to which Eve answers with the letter $u_j(i)$. 
If the two letters match the exchange is a success, otherwise it is a failure.

An alternating automaton cannot simulate this protocol, because it would need to choose $j \in \set{1,\ldots,k}$
at the beginning, even before reading $u$. The formal proof of this intuition is that this language
does not have subexponential alternating complexity, as proved in Subsection~\ref{subsec:lower_bound_example}.

However, if we swap the two inputs, i.e., the automaton reads $u_1 \sharp u_2 \sharp \cdots \sharp u_k$ before $u$,
then it can simulate the protocol: when reading $u_j$ it non-deterministically decides to store $u_j$,
and later checks using universal guesses that $u_j = u$.

This example shows that using alternating communication complexity would not yield strong lower bounds
for alternating state complexity. Building on the ideas behind the language $L$ one can obtain arbitrary gaps
between the two notions.

\section{A Lower Bound Technique}
\label{sec:lower_bound}

In this section, we develop a generic lower bound technique for alternating state complexity.
It is based on the size of generating families for some lattices of languages;
we describe it in Subsection~\ref{subsec:lattices}, and a concrete approach to use it,
based on query tables, is developed in Subsection~\ref{subsec:query}.
We apply it to an example in Subsection~\ref{subsec:lower_bound_example}.

\subsection{Boundedly Generated Lattices of Languages}
\label{subsec:lattices}
Let $L$ be a language and $u$ a word. 
Recall that the left quotient of $L$ with respect to $u$ is
$$u^{-1} L = \set{v \mid uv \in L}.$$
If $u$ has length at most $n$, we say that $u^{-1} L$ is a left quotient of $L$ of order $n$.

A lattice of languages is a set of languages closed under union and intersection.
Given a family of languages, the lattice it generates is the smallest lattice containing this family.

\begin{theorem}
\label{thm:lower_bound}
If $L$ is in $\Alt{f}$, then there exists a constant $C$ such that for all $n \in \N$,
there exists a family of at most $C \cdot f(n)$ languages whose generated lattice contains all the left quotients of $L$ of order $n$.
\end{theorem}

To some extent, Theorem~\ref{thm:lower_bound} draws from the classical Myhill-Nerode theorem~\cite{Nerode58}.
However, since there is no notion of minimal alternating automaton, the situation is more complicated here.
In particular, the converse of Theorem~\ref{thm:lower_bound} may not hold.

\vskip1em
Theorem~\ref{thm:lower_bound} reduces the question of finding lower bounds for alternating state complexity to the following one:
given a finite lattice of languages, what is the size of the smallest set of generators for this lattice?

\begin{proof}
Let $\A$ be an alternating automaton recognising $L$ of state complexity at most $f$.

Fix $n$. Let $Q_n$ denote the set of states reachable by some word of length at most $n$;
by assumption $|Q_n|$ is at most $C \cdot f(n)$ for some constant $C$.
For $q$ in $Q_n$, let $L(q)$ be the language recognised by $\A$ taking $q$ as initial state,
and $\L_n$ the family of these languages.

We prove by induction over $n$ that all left quotients of $L$ of order $n$ can be obtained as boolean combinations of languages in $\L_n$.

The case $n = 0$ is clear, since $\varepsilon^{-1} L = L = L(q_0)$.

Consider a word $w$ of length $n+1$, write $w = ua$.
We are interested in $w^{-1} L = a^{-1} (u^{-1} L)$, so let us start by considering $u^{-1} L$.
By the induction hypothesis, $u^{-1} L$ can be obtained as a boolean combination of languages in $\L_n$:
write $u^{-1} L = \phi(\L_n)$, meaning that $\phi$ is a boolean formula whose atoms are languages in $\L_n$.

Now consider $a^{-1} \phi(\L_n)$.
Observe that the left quotient operation respects both unions and intersections,
i.e., 
\[
a^{-1}(L_1 \cup L_2) = a^{-1} L_1 \cup a^{-1} L_2,
\]
and 
\[
a^{-1}(L_1 \cap L_2) = a^{-1} L_1 \cap a^{-1} L_2.
\]
It follows that $w^{-1} L = a^{-1} (\phi(\L_n)) = \phi(a^{-1} \L_n)$; 
this notation means that the atoms are languages of the form $a^{-1} M$ for $M$ in $\L_n$,
i.e., $a^{-1} L(q)$ for $q$ in $S_n$.

To finish the proof, we remark that $a^{-1} L(q)$ can be obtained as a boolean combination of the languages $L(p)$,
where $p$ are the states that appear in $\delta(q,a)$.
To be more precise, we introduce the notation $\psi(L(\cdot))$, on an example: if $\psi = p \wedge (r \vee s)$,
then $\psi(L(\cdot)) = L(p) \wedge (L(r) \vee L(s))$.
With this notation, $a^{-1} L(q) = \delta(a,q)(L(\cdot))$.
Thus, for $q$ in $Q_n$, we have that $a^{-1} L(q)$ can be obtained as a boolean combination of languages in $\L_{n+1}$.

Putting everything together, it implies that $w^{-1} L$ can be obtained as a boolean combination of languages in $\L_{n+1}$,
finishing the inductive proof.
\end{proof}

\subsection{The Query Table Method}
\label{subsec:query}
Thanks to Theorem~\ref{thm:lower_bound}, we are now looking at the size of the smallest set of generators for a given finite lattice of languages.
To study this quantity we define the notion of query tables.

\begin{definition}[Query Table]
Consider a family of languages $\L$.
Given a word $w$, its profile with respect to $\L$, or $\L$-profile,
is the boolean vector stating whether $w$ belongs to $L$, for each $L$ in $\L$.
The size of the query table of $\L$ is the number of different $\L$-profiles, when considering all words.

For a language $L$, its query table of order $n$ is the query table of the left quotients of $L$ of order $n$.
\end{definition}

The name query table comes from the following image, illustrated in Figure~\ref{fig:table}:
the query table of $\L$ is the infinite table whose columns are indexed by languages in $\L$ and rows by words (so, there are infinitely many rows).
The cell corresponding to a word $w$ and a language $L$ in $\L$ is the boolean indicating whether $w$ is in $L$.
Thus the $\L$-profile of $w$ is the row corresponding to $w$ in the query table of $\L$.

\begin{figure}[!ht]
\centering
\includegraphics[scale=.3]{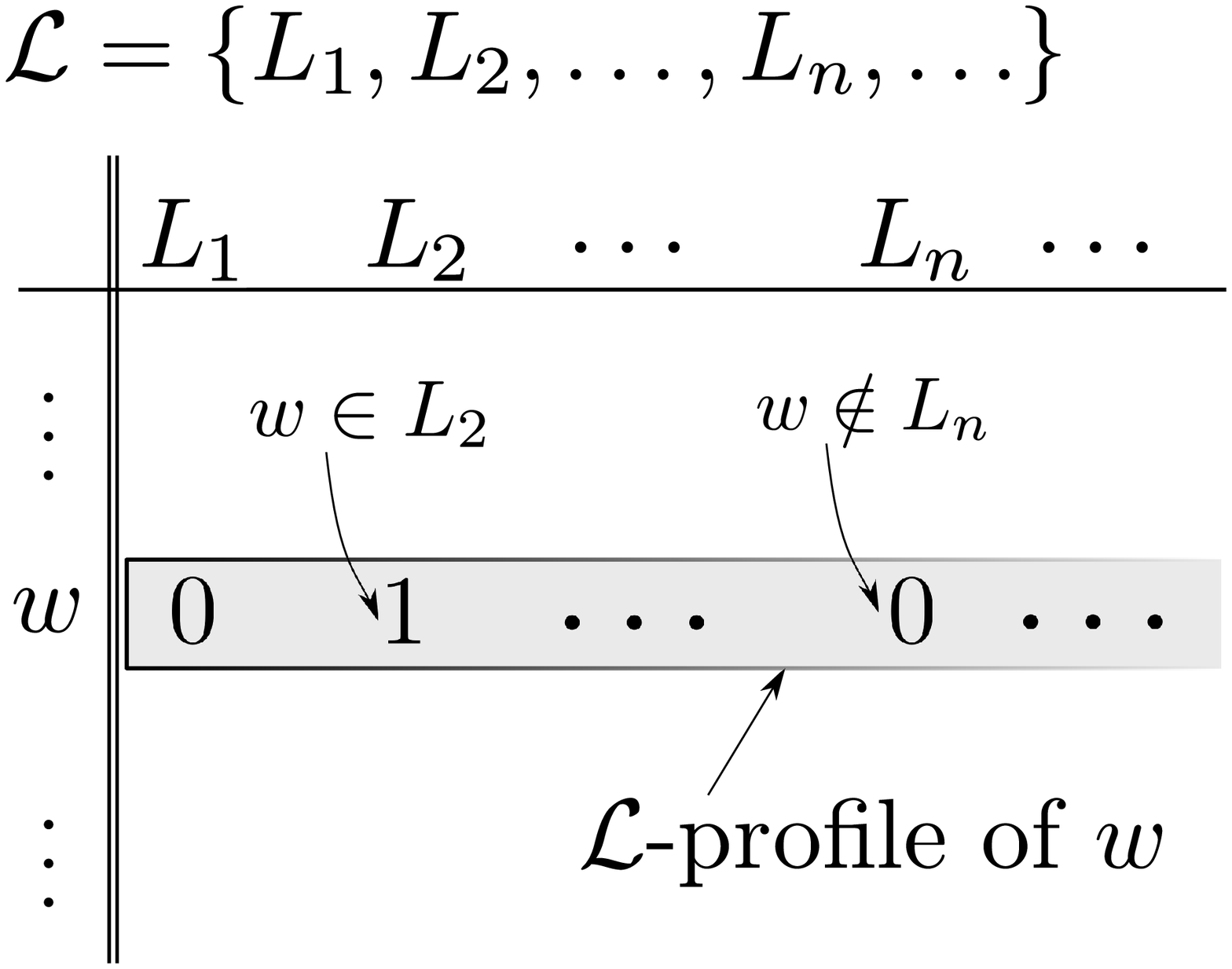}
\caption{The query table of $\L$.}
\label{fig:table}
\end{figure}

\begin{lemma}
\label{lem:upper_bound_query_table}
Consider a lattice of languages $\L$ generated by $k$ languages.
The query table of $\L$ has size at most $2^k$.
\end{lemma}

Indeed, there are at most $2^k$ different profiles with respect to $\L$.

\begin{theorem}
\label{thm:query_table}
Let $L$ in $\Alt{f}$.
There exists a constant $C$ such that for all $n \in \N$,
the query table of $L$ of order $n$ has size at most $2^{C \cdot f(n)}$.
\end{theorem}


The proof of Theorem~\ref{thm:query_table} relies on the following lemma.

\begin{lemma}
\label{lem:query_table}
Consider two families of languages $\L$ and $\M$.
If $\M \subseteq \L$, then the size of the query table of $\M$ is smaller than or equal to the size of the query table of $\L$.
\end{lemma}

\begin{proof}
It suffices to observe that the query table of $\M$ is ``included'' in the query table of $\L$.
More formally, consider in the query table of $\L$ the sub-table which consists of columns corresponding to languages in $\M$:
this is the query table of $\M$.
This implies the claim.
\end{proof}

We now prove Theorem~\ref{thm:query_table}.
Thanks to Theorem~\ref{thm:lower_bound}, the family of left quotients of $L$ of order $n$ is contained in 
a lattice generated by a family of size at most $C \cdot f(n)$.
It follows from Lemma~\ref{lem:query_table} that the size of the query table of $L$ of order $n$ 
is smaller than or equal to the size of the query table of a lattice generated by at most $C \cdot f(n)$ languages,
which by Lemma~\ref{lem:upper_bound_query_table} is at most $2^{C \cdot f(n)}$.

\vskip1em
Our lower bound apparatus is now complete: 
thanks to Theorem~\ref{thm:query_table}, 
to prove a lower bound on the alternating state complexity of a language $L$,
it is sufficient to prove lower bounds on the size of the query tables of $L$.

\subsection{A First Application of the Query Table Method}
\label{subsec:lower_bound_example}
As a first application of our technique, we exhibit a language which has asymptotically maximal (i.e., exponential) alternating state complexity.
Surprisingly, this language is simple in the sense that it is context-free and definable in Presburger arithmetic,
i.e., in first-order logic with the addition predicate.

Recall that $L$ has subexponential alternating state complexity if $L \in \Alt{f}$ for some $f$ such that $f = o(C^n)$ for all $C > 1$.
Thanks to Theorem~\ref{thm:query_table}, to prove that $L$ does not have subexponential alternating state complexity, it is enough to exhibit a constant $C > 1$ 
such that for infinitely many $n$, the query table of the left quotients of $L$ of order $n$ has size at least $2^{C^n}$.

\begin{theorem}\label{thm:exp_alt}
There exists a language which does not have subexponential alternating state complexity,
yet is both context-free and definable in Presburger arithmetic.
\end{theorem}

\begin{proof}
Let
$$L = \set{u \sharp u_1 \sharp u_2 \sharp \cdots \sharp u_k 
\left| 
\begin{array}{c}
u,u_1,\ldots,u_k \in \set{0,1}^*, \\
\exists j \in \set{1,\ldots,k}, u = u_j^R
\end{array}
\right.}.$$
Recall that the notation $u^R$ stands for the reverse of $u$.
Note, and this is very important here, the number of words $u_1,\ldots,u_k$ is not bounded: $k$ is arbitrary.

It is easy to see that $L$ is both context-free and definable in Presburger arithmetic,
i.e., in first-order logic with the addition predicate
(the use of reversed words in the definition of $L$ is only there to make $L$ context-free).

We show that $L$ does not have subexponential alternating state complexity.
We prove that for all $n$, the query table of the left quotients of $L$ of order $n$ has size at least $2^{2^n}$.
Thanks to Theorem~\ref{thm:query_table}, this implies the result.

Fix $n$. Let $U$ be the set of all words $u$ in $\set{0,1}^n$. It has cardinality $2^n$.
Consider a subset $S$ of $U$. 
We argue that there exists a word $w$ such that if $u$ is in $U$, then the following equivalence holds:
$$w \in u^{-1} L \Longleftrightarrow u \in S.$$
This shows the existence of $2^{2^n}$ different profiles with respect to the left quotients of order $n$, as claimed.

Let $u_1,\ldots,u_{|S|}$ be the words in $S$.
Consider 
$$w = \sharp u_1^R \sharp u_2^R \sharp \cdots \sharp u_{|S|}^R.$$
The word $w$ clearly satisfies the claim above.
\end{proof}

\section{A Hierarchy Theorem for Languages of Polynomial Alternating State Complexity}
\label{sec:hierarchy}

\begin{theorem}
For each $\ell \ge 2$, there exists a language $L_\ell$ such that:
\begin{itemize}
	\item $L_\ell$ is in $\Alt{n^\ell}$,
	\item $L_\ell$ is not in $\Alt{n^{\ell - \varepsilon}}$ for any $\varepsilon > 0$.
\end{itemize}
\end{theorem}

Consider the alphabet $\set{0,1} \cup \set{\lozenge,\sharp}$.

Let $\ell \ge 2$, and 
$$L_\ell = \set{\lozenge^p u \sharp u_1 \sharp u_2 \sharp \cdots \sharp u_k \left|
\begin{array}{c}
u,u_1,\ldots,u_k \in \set{0,1}^*,\\
k \le p^\ell, \exists j \le k,\ u = u_j
\end{array}\right.}.$$
We note that unlike the language used for proving Theorem~\ref{thm:exp_alt},
the value of $k$ is here bounded by $p^\ell$.

\begin{proof}
We construct an alternating automaton of state complexity $O(n^\ell)$.
The automaton has three consecutive phases:
\begin{enumerate}
	\item First, a non-deterministic guessing phase while reading $\lozenge^p$, which passes onto the second phase a number $j$ in $\set{1,\ldots,p^\ell}$.

	\vskip1em
	Formally, the set of states for this phase is $\N$, the initial state is $0$ and the transitions are
	$$\begin{array}{l}
	\delta(0,\lozenge) = 1 \\
	\delta(k^\ell,\lozenge) = \bigvee_{j \in \set{1,\ldots,(k+1)^\ell}} j \\
	\delta(p,\lozenge) = p.
	\end{array}$$
	
	The automaton for this phase has state complexity $n^\ell$.

	\item Second, a universal phase while reading $u$.
	For each $i$ in $\set{1,\ldots,|u|}$, the automaton launches one copy storing the position $i$, the letter $u(i)$ and the number $j$ guessed in the first phase.

	\vskip1em
	Formally, the set of states for this phase is \[\N \times (\set{0,1} \cup \set{\bot}) \times \N.\]
	The first component is the length of the word read so far (in this phase), the second component stores the letter read, 
	where the letter $\bot$ stands for undeclared, and the last component is the number $j$.

	The initial state is $(0,\bot,j)$.
	The transitions are
	$$\begin{array}{l}
	\delta((q,\bot,j),a) = (q+1,\bot,j) \wedge (q,a,j) \\
	\delta((q,a,j),b) = (q,a,j).
	\end{array}$$

	The automaton for this phase has quadratic state complexity.
	
	\item Third, a deterministic phase while reading \[\sharp u_1 \sharp u_2 \sharp \cdots \sharp u_k.\]
	It starts from a state of the form $(q,a,j)$.
	It checks whether $u_j(q) = a$.
	Localising $u_j$ is achieved by decrementing the number $j$ by one each time a letter $\sharp$ is read.
	In the corresponding $u_j$ localising the position $q$ is achieved by decrementing the first component by one at a time.

	The automaton for this phase has quadratic state complexity.
\end{enumerate}

We now prove the lower bound.

We prove that for all $n$, the size of the query table of $L_\ell$ of order $n + 2^{\frac{n}{\ell}}$ is at least $2^{2^n}$.
Thanks to Theorem~\ref{thm:query_table}, this implies that $L_\ell$ is not in $\Alt{n^{\ell - \varepsilon}}$ for any $\varepsilon > 0$.

Fix $n$. 
Let $U$ be the set of all words $u$ in $\set{0,1}^n$. It has cardinality $2^n$.

Observe that $\lozenge^{2^{\frac{n}{\ell}}} u \sharp u_1 \sharp u_2 \sharp \cdots \sharp u_{2^n}$ belongs to $L_\ell$
if and only if there exists $j$ in $\set{1,\ldots,2^n}$ such that $u = u_j$.

Consider any subset $S$ of $U$, we argue that there exists a word $w$ which
satisfies that if $u$ is in $U$, then the following equivalence holds:
$$w \in \left( \lozenge^{2^{\frac{n}{\ell}}}u \right)^{-1} L \Longleftrightarrow u \in S.$$
This shows the existence of $2^{2^n}$ different profiles with respect to the left quotients of order $n + 2^{\frac{n}{\ell}}$, as claimed.

Let $u_1,\ldots,u_{|S|}$ be the words in $S$.
Consider 
$$w = \sharp u_1 \sharp u_2 \sharp \cdots \sharp u_{|S|}.$$
The word $w$ clearly satisfies the claim above.
\end{proof}

\section{The Alternating State Complexity of Prime Numbers}
\label{sec:prime}

In this section, we give lower bounds on the alternating state complexity of the language of prime numbers written in binary:
$$\Prime = \set{u \in \set{0,1}^* \mid \bin(u) \textrm{ is prime}}.$$
By definition $\bin(w) = \sum_{i \in \set{0,\ldots,n-1}} w(i) 2^i$;
note that the least significant digit is on the left.

The complexity of this language has long been investigated; many efforts have been put in finding upper and lower bounds.
In 1976, Miller gave a first conditional polynomial time algorithm, assuming the generalised Riemann hypothesis~\cite{Miller76}.
In 2002, Agrawal, Kayal and Saxena obtained the same results, but non-conditional, 
i.e., not predicated on unproven number-theoretic conjectures~\cite{AKS02}.

\vskip1em
The first lower bounds were obtained by Hartmanis and Shank in 1968,
who proved that checking primality requires at least logarithmic deterministic space~\cite{HartmanisShank68},
conditional on number-theoretic assumptions.
It was shown by Hartmanis and Berman in 1976 that if the number is presented in unary, 
then logarithmic deterministic space is necessary and sufficient~\cite{HartmanisBerman76}.
The best lower bound from circuit complexity is due to Allender, Saks and Shparlinski: they proved unconditionally in 2001 that $\Prime$ 
is not in $\mathrm{AC}^0[p]$ for any prime $p$~\cite{ASS01}.

\vskip1em
The results above are incomparable to our setting, as we are here interested in state complexity.
The first and only result to date about the state complexity of $\Prime$ is due to Hartmanis and Shank in 1969:

\begin{theorem}[\cite{HartmanisShank69}]
\label{thm:hs}
The set of prime numbers written in binary does not have subexponential deterministic state complexity.
\end{theorem}

Their result is unconditional, and makes use of Dirichlet's theorem on arithmetic progressions of prime numbers.
A related and stronger result has been proved by Shallit~\cite{Shallit96}, 
which says that the deterministic automaticity of the prime numbers is not subexponential.

Hartmanis and Shank proved the following result.

\begin{lemma}[\cite{HartmanisShank69}]
\label{lemma:hs}
Fix $n > 1$, and consider $u$ and $v$ two different words of length $n$ starting with a $1$.
Then the left quotients $u^{-1} \Prime$ and $v^{-1} \Prime$ are different.
\end{lemma}

Lemma~\ref{lemma:hs} directly implies Theorem~\ref{thm:hs}~\cite{HartmanisShank69}.
It also yields a lower bound of $n-1$ on the size of the query table of $\Prime$ of order $n$.
Thus, together with Theorem~\ref{thm:query_table}, this proves that $\Prime$ does not have sublogarithmic alternating state complexity.

\begin{corollary}
The set of prime numbers written in binary does not have sublogarithmic alternating state complexity.
\end{corollary}

Our contribution in this section is to extend this result by showing that $\Prime$ does not have sublinear alternating state complexity,
which is an exponential improvement.

\begin{theorem}
\label{thm:prime}
The set of prime numbers written in binary does not have sublinear alternating state complexity.
\end{theorem}

Our result is unconditional, but it relies on the following advanced theorem from number theory, 
which can be derived from the results obtained by Maier and Pomerance~\cite{MP90}.
Note that their results are more general; we state a corollary fitting our needs.
Simply put, this result says that in any (reasonable) arithmetic progression and for any $d$, 
there exists a prime number in this progression at distance at least $d$ from all other prime numbers.

\begin{theorem}[\cite{MP90}]
\label{thm:maier_pomerance}
For every arithmetic progression $a + b \N$ such that $a$ and $b$ are coprime, for every $N$,
there exists a number $k$ such that $p = a + b \cdot k$ is the only prime number in $[p-N,p+N]$.
\end{theorem}

We proceed to the proof of Theorem~\ref{thm:prime}.

\begin{proof}
We show that for all $n > 1$, the query table of $\Prime$ of order $n$ has size at least $2^{n-1}$.
Thanks to Theorem~\ref{thm:query_table}, this implies the result.

Fix $n > 1$. Let $U$ be the set of all words $u$ of length $n$ starting with a $1$. 
Equivalently, we see $U$ as a set of numbers; it contains all the odd numbers smaller than $2^n$.
It has cardinality $2^{n-1}$.

We argue that for all $u$ in $U$, there exists a word $w$ such that for all $v$ in $U$, 
$w$ is in $v^{-1} \Prime$ if and only if $u = v$.
In other words the profile of $w$ is $0$ everywhere but on the column $u^{-1} \Prime$.
Let $u$ in $U$; write $a = \bin(u)$.
Consider the arithmetic progression $a + 2^n \N$; note that $a$ and $2^n$ are coprime.
Thanks to Theorem~\ref{thm:maier_pomerance}, for $N = 2^n$, there exists 
a number $k$ such that $p = a + 2^n \cdot k$ is the only prime number in $[p-N,p+N]$.
Let $w$ be a word such that $\bin(w) = k$.
We show that for all $v$ in $U$, we have the following equivalence: $w$ is in $v^{-1} \Prime$ if and only if $u = v$.

Indeed, $\bin(vw) = \bin(v) + 2^n \cdot \bin(w)$. 
Observe that 
$$|\bin(vw) - \bin(uw)| = |\bin(v) - \bin(u)| < 2^n.$$
Since $p$ is the only prime number in $[p-2^n,p+2^n]$, the equivalence follows.

We constructed $2^{n-1}$ words each having a different profile, implying the claimed lower bound.
\end{proof}

Theorem~\ref{thm:prime} proves a linear lower bound on the alternating state complexity of $\Prime$.
We do not know of any non-trivial upper bound, and believe that there are none, meaning that $\Prime$ does not have subexponential alternating state complexity.

An evidence for this is the following probabilistic argument. Consider the distribution of languages over $\set{0,1}^*$
such that a word $u$ in thrown into the language with probability $\frac{1}{|u|}$.
It is a common (yet flawed) assumption that the prime numbers satisfy this distribution, as witnessed for instance by the prime number theorem.
One can show that with high probability such a language does not have subexponential alternating state complexity,
the reason being that two different words are very likely to induce different profiles in the query table.
Thus it is reasonable to expect that $\Prime$ does not have subexponential alternating state complexity.

\vskip1em
We dwell on the possibility of proving stronger lower bounds for the alternating state complexity of $\Prime$.
Theorem~\ref{thm:maier_pomerance} fleshes out the \textit{sparsity} of prime numbers:
it constructs isolated prime numbers in any arithmetic progression,
and allows us to show that the query table of $\Prime$ contains all profiles with all but one boolean value set to false.

To populate the query table of $\Prime$ further, one needs results witnessing the \textit{density} of prime numbers,
i.e., to prove the existence of clusters of prime numbers.
This is in essence the contents of the Twin Prime conjecture, or more generally of Dickson's conjecture,
which are both long-standing open problems in number theory,
suggesting that proving better lower bounds is a very challenging objective.
Dickson's conjecture reads (we use the equivalent statement given by Ribenboim in~\cite{Ribenboim96}, called $D_1$):

\begin{conjecture}[Dickson's Conjecture]
\label{conjecture:dickson}
Fix $b$ and \[S = \set{1 \le a_1 < \cdots < a_s < b}\] such that 
there exists no prime number $p$ which divides \[\prod_{a \in S} (b \cdot k + a)\] for every $k$ in $\N$.
Then there exists a number $k$ such that 
\[b \cdot k + a_1, b \cdot k + a_2, \ldots, b \cdot k + a_s\] are consecutive prime numbers.
\end{conjecture}

\begin{theorem}
Assuming Conjecture~\ref{conjecture:dickson} holds true, 
the set of prime numbers written in binary does not have subexponential alternating state complexity.
\end{theorem}

\begin{proof}
We show that for infinitely many $n > 1$, the query table of $\Prime$ of order $n$ has size doubly-exponential in $n$.
Thanks to Theorem~\ref{thm:query_table}, this implies the result.

Fix $n > 1$. As above, let $U$ be the set of all words $u$ of length $n$ starting with a $1$, i.e., odd numbers.
For a subset \[S = \set{1 \le a_1 < \cdots < a_s < b}\] of $U$, let $(\lozenge)$ denote the property that 
there exists no prime number~$p$ which divides $\prod_{a \in S} (b \cdot k + a)$ for every $k$ in $\N$.

Let $S$ be a subset of $U$ satisfying $(\lozenge)$.
Thanks to Conjecture~\ref{conjecture:dickson}, there exists a number $k$ such that 
for $a_1 \le a \le a_s$, the number $2^n \cdot k + a$ is prime if and only if $a$ is in $S$.
Let $w$ be a word such that $\bin(w) = k$. 
It clearly satisfies the condition above.
In other words the profile of $w$ for the columns between $a_1$ and $a_s$ is $1$ on the columns corresponding to $S$, and $0$ everywhere else.
For each subset $S$ satisfying $(\lozenge)$ with the same extremal elements ($a_1$ and~$a_s$) 
we constructed a word such that these words have pairwise different profiles.

To finish the proof, we need to explain why this induces doubly-exponentially many different profiles.
For any $n$, the set $S$ of odd numbers $a \in U$ such that $2^n + a$ is a prime number satisfies $(\lozenge)$.
This follows from the remark that no prime number can divide both $\prod_{a \in S} a$ and $\prod_{a \in S} (2^n + a)$.
Thanks to the prime number theorem estimating the proportion of prime numbers, 
we know that for infinitely many $n$ the set $S$ contains a number $a_1$ smaller than $2^{n-2}$ and a number $a_s$ larger than $2^n - 2^{n-2}$.
Now, each subset of $S$ gives rise to a different profile, which yields doubly-exponentially many of them.
\end{proof}

\section*{Conclusion}

Our first result is to show that probabilistic languages can have arbitrarily high deterministic state complexity,
substantiating a claim by Rabin.
Our main technical contributions concerns the alternating state complexity, for which we have developed a generic lower bound technique 
and applied it to two problems.
The first result is to give languages of arbitrary high polynomial alternating state complexity.
The second result is to give lower bounds on the alternating state complexity of the language of prime numbers;
we show that it is not sublinear, which is an exponential improvement over the previous result.
However, the exact complexity is left open; we conjecture that it is not subexponential, but obtaining this result might require
major advances in number theory.

\vskip1em
We leave three questions open, motivating further research:
\begin{itemize}
	\item What is the alternating state complexity of probabilistic languages? We conjecture that the probabilistic language we introduced does not have subexponential alternating state complexity, but our lower bound technique does not suffice to prove this result.
	
	\item Is the converse of Theorem~\ref{thm:lower_bound} true, or in other words does the size of the query table
	completely characterise the alternating state complexity (as it does in the deterministic case)?
	We believe the answer is ``no'', but proving it would require using a stronger lower bound technique
	to separate alternating state complexity from size of the query table.
	
	\item Can we find a notion of reduction between languages which respects the alternating state complexity,
	inducing a definition of completeness for alternating state complexity classes?
	The sequence of languages $L_\ell$ for $\ell \ge 2$ are good candidates for complete languages in the polynomial hierarchy.
\end{itemize}

\section*{References}
\bibliography{bib}

\begin{thebibliography}{10}
\expandafter\ifx\csname url\endcsname\relax
  \def\url#1{\texttt{#1}}\fi
\expandafter\ifx\csname urlprefix\endcsname\relax\def\urlprefix{URL }\fi
\expandafter\ifx\csname href\endcsname\relax
  \def\href#1#2{#2} \def\path#1{#1}\fi

\bibitem{Fijalkow16}
N.~Fijalkow, The online space complexity of probabilistic languages, in:
  LFCS'2016, 2016, pp. 1--12.

\bibitem{Fijalkow18}
N.~Fijalkow, \href{https://doi.org/10.1145/3209108.3209167}{The state
  complexity of alternating automata}, in: LICS'18, 2018, pp. 414--421.
\newblock \href {http://dx.doi.org/10.1145/3209108.3209167}
  {\path{doi:10.1145/3209108.3209167}}.
\newline\urlprefix\url{https://doi.org/10.1145/3209108.3209167}

\bibitem{Karp67}
R.~M. Karp, Some bounds on the storage requirements of sequential machines and
  {T}uring machines, Journal of the {ACM} 14~(3).

\bibitem{Nerode58}
A.~Nerode, Linear automaton transformations, Proceedings of the American
  Mathematical Society 9~(4) (1958) 541--544.

\bibitem{HartmanisShank69}
J.~Hartmanis, H.~Shank, Two memory bounds for the recognition of primes by
  automata, Mathematical Systems Theory 3~(2).

\bibitem{Rabin63}
M.~O. Rabin, Probabilistic automata, Information and Control 6~(3) (1963)
  230--245.

\bibitem{ChandraStockmeyer76}
A.~K. Chandra, L.~J. Stockmeyer, Alternation, in: FOCS'76, 1976, pp. 1--12.

\bibitem{Kozen76}
D.~Kozen, On parallelism in {T}uring machines, in: FOCS'76, 1976, pp. 89--97.

\bibitem{ChandraKozenStockmeyer81}
A.~K. Chandra, D.~Kozen, L.~J. Stockmeyer, Alternation, Journal of the {ACM}
  28~(1) (1981) 114--133.

\bibitem{Reiter15}
F.~Reiter, Distributed graph automata, in: LICS, 2015, pp. 1--12.

\bibitem{Yu01}
S.~Yu, State complexity of regular languages, Journal of Automata, Languages
  and Combinatorics 6~(2) (2001) 221.

\bibitem{Yu02}
S.~Yu, State complexity of finite and infinite regular languages, Bulletin of
  the {EATCS} 76 (2002) 142--152.

\bibitem{GMRY17}
Y.~Gao, N.~Moreira, R.~Reis, S.~Yu, A survey on operational state complexity,
  Journal of Automata, Languages and Combinatorics 21~(4) (2017) 251--310.

\bibitem{DK07}
A.~Dawar, S.~Kreutzer, Generalising automaticity to modal properties of finite
  structures, Theoretical Computer Science 379~(1-2) (2007) 266--285.

\bibitem{FS15}
N.~Fijalkow, M.~Skrzypczak, Irregular behaviours for probabilistic automata,
  in: RP'2015, 2015, pp. 33--36.

\bibitem{ShallitBreitbart96}
J.~Shallit, Y.~Breitbart, Automaticity {I:} properties of a measure of
  descriptional complexity, Journal of Computer and System Sciences 53~(1)
  (1996) 10--25.

\bibitem{PomeranceRobsonShallit97}
C.~Pomerance, J.~M. Robson, J.~Shallit, Automaticity {II:} descriptional
  complexity in the unary case, Theoretical Computer Science 180~(1-2) (1997)
  181--201.

\bibitem{GlaisterShallit98}
I.~Glaister, J.~Shallit, Automaticity {III:} polynomial automaticity and
  context-free languages, Computational Complexity 7~(4) (1998) 371--387.

\bibitem{Shallit96}
J.~Shallit, Automaticity {IV:} sequences, sets, and diversity, Journal de
  Th{\'e}orie des Nombres de Bordeaux 8~(2) (1996) 347--367.

\bibitem{B63}
J.~Brzozowski, Canonical regular expressions and minimal state graphs for
  definite events, Symposium on Mathematical Theory of Automata 12 (1963)
  529--561.

\bibitem{RS97}
G.~Rozenberg, A.~Salomaa, Handbook of Formal Languages, Springer-Verlag Berlin
  Heidelberg, Springer, 1997.

\bibitem{FJY90}
A.~Fellah, H.~J{\"u}rgensen, S.~Yu, Constructions for alternating finite
  automata, International Journal of Computer Mathematics 35~(1) (1990)
  117--132.

\bibitem{BabaiFranklSimon86}
L.~Babai, P.~Frankl, J.~Simon, Complexity classes in communication complexity
  theory (preliminary version), in: FOCS'86, 1986, pp. 1--12.

\bibitem{Miller76}
G.~L. Miller, Riemann's hypothesis and tests for primality, Journal of Computer
  and System Sciences 13~(3) (1976) 300--317.

\bibitem{AKS02}
M.~Agrawal, N.~Kayal, N.~Saxena, Primes is in {P}, Annals of Mathematics 2
  (2002) 781--793.

\bibitem{HartmanisShank68}
J.~Hartmanis, H.~Shank, On the recognition of primes by automata, Journal of
  the {ACM} 15~(3) (1968) 382--389.

\bibitem{HartmanisBerman76}
J.~Hartmanis, L.~Berman, On tape bounds for single letter alphabet language
  processing, Theoretical Computer Science 3~(2) (1976) 213--224.

\bibitem{ASS01}
E.~Allender, M.~E. Saks, I.~Shparlinski, A lower bound for primality, Journal
  of Computer and System Sciences 62~(2) (2001) 356--366.

\bibitem{MP90}
H.~Maier, C.~Pomerance, Unusually large gaps between consecutive primes,
  Transactions of the American Mathematical Society 322~(1) (1990) 201--237.

\bibitem{Ribenboim96}
P.~Ribenboim, The Book of Prime Number Records, Discrete Mathematics,
  Springer-Verlag New York, 1996.

\end{thebibliography}

\end{document}